%% file: main.tex
\documentclass[envcountsame]{llncs}

\usepackage[english]{babel}
\usepackage[mathletters]{ucs}
\usepackage[utf8x]{inputenc} 

\usepackage{caption}
\usepackage{subcaption}
\usepackage{tikz}

\usepackage{enumerate}
\usepackage{amsmath}
\usepackage{amsfonts}

\usepackage{macros}

\title{
	Strong Normalizability as a Finiteness Structure via the Taylor Expansion of λ-terms
	\thanks{Supported by French ANR Project Coquas (number ANR 12 JS02 006 01).}
}
\titlerunning{Strong Normalizability as a Finiteness Structure}
\author{Michele Pagani\inst{1}  \and Christine Tasson\inst{1}  \and Lionel Vaux\inst{2}}
\institute{
	Université Paris Diderot, CNRS, IRIF UMR 8243, F-75205, Paris, France
	\email{\{christine.tasson, michele.pagani\}@pps.univ-paris-diderot.fr}
	\and
	Aix Marseille Université, CNRS, Centrale Marseille, I2M UMR 7373, F-13453, Marseille, France\\
	\email{lionel.vaux@univ-amu.fr}
}

\begin{document}
\maketitle

\begin{abstract}
In the folklore of linear logic, a common intuition is that the structure of finiteness spaces, introduced by Ehrhard, semantically reflects the strong normalization property of cut-elimination.

We make this intuition formal in the context of the non-deterministic $\lambda$-calculus by introducing a finiteness structure on resource terms, which is such that a $\lambda$-term is strongly normalizing iff the support of its Taylor expansion is finitary. 


An application of our result is the existence of a normal form for the Taylor expansion of any strongly normalizable non-deterministic $\lambda$-term. 
\end{abstract}

\section{Introduction}
\input{A-intro}

\section{Preliminaries}\label{sect:preliminary}
\input{B-preliminaries}

\section{Strongly Normalizing Terms Are $\mathcal D_+$ Typable}\label{sect:D+}
\input{C-systemD}

\section{$\mathcal D_+$ Typable Terms Are Finitary}\label{sect:realizers}
\input{D-sn_fin}

\section{Finitary Terms Are Strongly Normalizing}\label{sect:fin_to_sn}
\input{E-fin_sn}

\section{Conclusion}\label{sect:conclusion}
\input{F-conclusion}

\bibliographystyle{splncs}
\bibliography{bib}

\end{document}

%% file: A-intro.tex

%
It is well-known that sets and relations can be presented as a category of modules and linear functions over the boolean semi-ring, giving one of the simplest semantics of linear logic. In~\cite{lamarche92} (see also~\cite{LairdMMP13}), it is shown how to generalize this construction to any complete\footnote{A semi-ring is complete if any sum, even infinite, is well-defined.} semi-ring $\mathcal R$ and yet obtain a model of linear logic. In particular, the composition of two matrices $\phi\in\mathcal R^{A\times B}$ and $\psi\in\mathcal R^{B\times C}$ is given by the usual matrix multiplication:
\begin{equation}\label{eq:weighted_composition}
(\phi ; \psi)_{a,c} \eqdef \sum_{b\in B} \phi_{a,b}\cdot\psi_{b,c}
\end{equation}
The semi-ring $\mathcal R$ must be complete because the above sum might be infinite. This is an issue, because it prevents us from considering standard vector spaces, which are usually constructed over ``non-complete'' fields like reals or complexes.

In order to overcome this problem, Ehrhard introduced the notion of finiteness space \cite{ehrhardfs}. A \emph{finiteness space} is a pair of a set $A$ and a set $\fA$ (called \emph{finiteness structure} in Definition~\ref{def:finiteness}) of subsets of $A$ which is closed under a notion of duality. The point is that, for any field $\mathcal K$ (resp. any semi-ring), the set of vectors in $\mathcal K^{A}$ whose support\footnote{The support of $v\in\mathcal K^{A}$ is the set of those $a\in A$ such that the scalar $v_a$ is non-null.} is in $\fA$ constitutes a vector space (resp. a module) over $\mathcal K$. Moreover, any two matrices $\phi\in\mathcal K^{A\times B}$ and $\psi\in\mathcal K^{B\times 
C}$ whose supports are in resp.\ $\fA\multimap\fB$ and $\fB\multimap\fC$ (the finiteness structures associated with the linear arrow) compose, because the duality condition on the supports makes the terms in the sum of Equation~\eqref{eq:weighted_composition} be zero almost everywhere. This gives rise to a category which is a model of linear logic and its differential extension. 

The notion of finiteness space seems strictly related to the property of normalization. Already in~\cite{ehrhardfs}, it is remarked that the coKleisli category of the exponential comonad is a model of simply typed $\lambda$-calculus, but it is not cpo-enriched and thus cannot interpret (at least in a standard way) fixed-point combinators, so neither PCF nor untyped $\lambda$-calculus. Moreover, in the setting of differential nets, Pagani showed that the property of having a finitary interpretation corresponds to an acyclicity criterion (called \emph{visible acyclicity} \cite{Pag08}) which entails the normalization property of the cut-elimination procedure \cite{Pag09tlca}, while there are examples of visibly cyclic differential nets which do not normalize. 

The goal of this paper is to shed further light on the link between finiteness spaces and normalization, this time considering the non-deterministic untyped $\lambda$-calculus. Since we deal with $\lambda$-terms and not with linear logic proofs (or differential nets), we will speak about formal power series rather than matrices at the semantical level. This corresponds to move from the morphisms of a linear category to those of its coKleisli construction. Moreover, following~\cite{difftaylor}, we describe the monomials of these power series as \emph{resource terms} in normal form. 
 The benefit of this setting is that the interpretation of a $\lambda$-term $M$ as a power series $[|M|]$ can be decomposed in two distinct steps: first, the term $M$ is associated with a formal series $\taylorExp M$ of resource terms possibly with redexes, called the \emph{Taylor expansion} of $M$ (see Table~\ref{table:taylor}); second, one reduces each resource term $t$ appearing in the support $\taylor M$ of $\taylorExp M$ into a normal form $\NF(t)$ and sum up all the results, that is ($\taylorExp M_t$ denotes the coefficient of $t$ in $\taylorExp M$):
\begin{equation}\label{eq:power_series_taylor}
 [|M|] = \sum_{t\in\taylor M} \taylorExp M_t\cdot\NF(t)
\end{equation}

The issue about the convergence of infinite sums appears in Equation~\eqref{eq:power_series_taylor} because there might be an infinite number of resource terms in $\taylor M$ reducing to the same normal form and thus possibly giving infinite coefficients to the formal series $[|M|]$. Ehrhard and Regnier proved in \cite{difftaylor} that this is not the case for deterministic $\lambda$-terms, however the situation gets worse in presence of non-deterministic primitives. If we allow sums $M+N$ representing potential reduction to $M$ or $N$, then 
one can construct terms evaluating to a variable $y$ an infinite number of times, like (where $\mathbf\Theta$ denotes Turing's fixed-point combinator):
\begin{equation}\label{eq:infiniteterm}
\appl{\mathbf\Theta}{\lambda x.(x+y)}\betaRed^\ast\appl{\mathbf\Theta}{\lambda x.(x+y)}+y\betaRed^\ast\appl{\mathbf\Theta}{\lambda x.(x+y)}+y+y\betaRed^\ast\dots
\end{equation}
We postpone to Examples~\ref{ex:lambda-terms} and \ref{ex:undefined_nf} a more detailed discussion of $\appl{\mathbf\Theta}{\lambda x.(x+y)}$, however we can already guess that this interplay between infinite reductions and non-determinism may  produce infinite coefficients. 

One can then wonder whether there are interesting classes of terms where the coefficients of the associated power series can be kept finite. Ehrhard proved in~\cite{Ehrhard10lics} that the terms typable by a non-deterministic variant of Girard's System $\mathcal F$ have always finite coefficients. A by-product of our results is Corollary~\ref{cor:NF}, which is a generalization of Ehrhard's result: every strongly normalizable non-deterministic $\lambda$-term can be interpreted by a power series with finite coefficients. 

The main focus of our paper is however in the means used for obtaining this result rather than on the result itself. The proof in~\cite{Ehrhard10lics} is based on a finiteness structure $\fS$ over the set of resource terms $\resourceTerms$, such that any term $M$ typable in System $\mathcal F$ has the support $\taylor M$ of its Taylor expansion in $\fS$. 
 We show that this method can be both generalized and strengthened in order to characterize strong normalization via finiteness structures. Namely, we give sufficient conditions on a finiteness structure $\fS$ over $\resourceTerms$ such that for every non-deterministic $\lambda$-term $M$: (i) if $M$ is strongly normalizable, then $\taylor M\in\fS$ (Corollary~\ref{cor:realizability}); (ii) if $\taylor M\in\fS$, then $M$ is strongly normalizable (Theorem~\ref{theorem:finitary:implies:SN}). 
 
 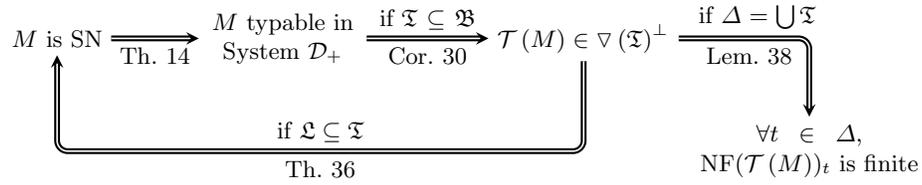
\begin{figure}[t]
\centering
\begin{tikzpicture}
\node (sn) at (0,0) {$M$ is SN};
\node[text width=2cm, text centered] (d) at (3,0) {$M$ typable in\\ System $\mathcal D_+$};
\node (fin) at (7,0) {$\taylor{M}\in\finitarySets[\fT] $};
\node[text width = 4cm, text centered] (taylor) at (10,-1.5) {$\forall t\in\Delta$,\\ $\NF(\taylor{M})_t$ is finite};
\draw[thick,->,>=stealth,double] (sn) --node [below, font=\small] {Th.~\ref{theorem:sn:typable}} (d);
\draw[thick,->,>=stealth,double] (d) --node [below, font=\small] {Cor.~\ref{cor:realizability}} node [above, font=\small, text width=1.7cm, text centered] {if $\fT\subseteq\boundedStructure$
}(fin);
\draw[thick,->,>=stealth,double,rounded corners] (fin) 
|- 
node {} (3.5,-1.5) 
node[above, font=\small] {if $\linearStructure\subseteq\fT$}
node[below, font=\small] {Th.~\ref{theorem:finitary:implies:SN}} 
-| (sn);
\draw[thick,->,>=stealth,double,rounded corners] (fin)  -| 
node [above, font=\small, pos=.3] {if $\resourceTerms=\bigcup\fT$}
node [below, font=\small, pos=.27] {Lem.~\ref{lemma:NF}}  
(taylor);
\end{tikzpicture}
\caption{main results of the paper}\label{fig:main_results}
\end{figure}
\paragraph{Contents. } Section~\ref{sect:preliminary} gives the preliminary definitions: the non-deterministic $\lambda$-calculus, its Taylor expansion into formal series of resource terms and the notion of finiteness structure. The proof of Item (i) splits into Sections~\ref{sect:D+} and~\ref{sect:realizers}, using an intersection type system (Table~\ref{fig:D+}) for characterizing strong normalization. Section~\ref{sect:fin_to_sn} gives the proof of Item (ii) and Section~\ref{sect:conclusion} concludes with Corollary~\ref{cor:NF} about the finiteness of the coefficients of the power series of strongly normalizable terms. Figure~\ref{fig:main_results} sums up the main results of the paper.

%% file: B-preliminaries.tex
%

\subsection{Non-deterministic $\lambda$-calculus $\ndTerms$}
\label{subsec:NDL}
The non-deterministic $\lambda$-calculus is defined by the following grammar\footnote{We use Krivine's notation with the standard conventions, see \cite{lcalculen} for reference.}
:
\begin{align*}
\text{$\lambda$-terms }&\ndTerms:& M \eqdef x \mid \labs x M \mid \appl M M \mid M + M
\end{align*}
subject to $\alpha$-equivalence and to the following identities:
%
\[
M+N =  N+M, \qquad (M+N)+P = M+(N+P),
\]
\vspace{-.3cm}
\[
\labs x {(M+N)} =  \labs xM + \labs xN,
\qquad
\appl {M+N}P =  \appl MP + \appl NP.
\]
The last two equalities state that abstraction is a linear operation (i.e.\ commutes with sums) while application is linear only in the function but not in the argument (i.e.~$\appl P{(M+N)}\neq  \appl PM + \appl PN$). Notice also that the sum is not idempotent: $M+M\neq M$. This is a crucial feature for making a difference between terms reducing to a value once, twice, more times or an infinite number of times (see the discussion about Equation~\eqref{eq:infiniteterm} in the Introduction). 

Although this follows an old intuition from linear logic, the first extension of the $\lambda$-calculus with (a priori non idempotent) sums subject to these identities was, as far as we know, the differential $\lambda$-calculus of Ehrhard and Regnier \cite{lambdadiff}. This feature is now quite standard in the literature following this revival of quantitative semantics.


The (capture avoiding) substitution of a term for a variable is defined as
usual and $\beta$-reduction $\betaRed$ is defined as the context closure of:
\[ \appl {\labs xM} N \betaRed \subst MxN.\]
We denote as $\betaRed^*$ the reflexive-transitive closure of $\betaRed$. 

\begin{example}\label{ex:lambda-terms}
Some $\lambda$-terms which will be used in the following are:
\[
\mathbf\Delta\eqdef\lambda x.\appl xx,\qquad
\mathbf\Omega\eqdef\appl{\mathbf\Delta}{\mathbf\Delta},\qquad
\mathbf{\Delta_3}\eqdef\lambda x.\appl x{xx},
\qquad\mathbf{\Omega_3}\eqdef\appl{\mathbf{\Delta_3}}{\mathbf{\Delta_3}},
\]
\[
\mathbf\Theta\eqdef\appl{\lambda xy. \appl{y}{\appl x{xy}}}{\lambda xy. \appl{y}{\appl x{xy}}}.   
\]
The term $\mathbf\Omega$ is the prototypical diverging term, reducing to itself in one single $\beta$-step. $\mathbf{\Omega_3}$ is another example of diverging term, producing terms of greater and greater size: 
$\mathbf{\Omega_3}\betaRed\appl{\mathbf{\Omega_3}}{\mathbf{\Delta_3}}\betaRed\appl{\mathbf{\Omega_3}}{\mathbf{\Delta_3}\mathbf{\Delta_3}}\betaRed\dots$
It will be used in Remark~\ref{rk:counterexample_singletons} to prove the subtlety of characterizing strong normalization with finiteness spaces. 
The Turing fixed-point combinator $\mathbf\Theta$ has been used in the Introduction to construct $\appl{\mathbf\Theta}{\lambda x.(x+y)}$ as an example of non-deterministic $\lambda$-term morally reducing to normal forms with infinite coefficients (Equation~\eqref{eq:infiniteterm}).
 Notice that, by abstraction linearity, $\appl{\mathbf\Theta}{\lambda x.(x+y)}=\appl{\mathbf\Theta}{(\lambda x.x+\lambda x.y)}$, however $\appl{\mathbf\Theta}{(\lambda x.x+\lambda x.y)}\neq\appl{\mathbf\Theta}{\lambda x.x}+\appl{\mathbf\Theta}{\lambda x.y}$, because application is not linear in the argument. This distinction is crucial: the latter term reduces to $(\appl{\mathbf\Theta}{\lambda x.x}) + y$,  with $\appl{\mathbf\Theta}{\lambda x.x}$ reducing to itself without producing any further occurence of $y$.
\end{example}

\subsection{Resource calculus $\resourceTerms$ and Taylor expansion}
\paragraph{The syntax.} \emph{Resource terms} and \emph{bags} are given by  mutual induction:
\begin{align*}
\text{resource terms }&\resourceTerms:&	s & \eqdef x \mid \lambda x.s \mid \rappl s{\ms s}&&&
\text{bags }&\bags:&\ms s & \eqdef  1 \mid s\cdot \ms s
\end{align*}
subject to both $\alpha$-equivalence and permutativity of $(\cdot)$: we most
often write $[s_1,\dotsc,s_n]$ for $s_1\cdot (\cdots \cdot (s_n \cdot 1)\cdots)$ and then
$[s_1,\dotsc,s_n]= [s_{\sigma(1)},\dotsc,s_{\sigma(n)}]$ for any permutation
$\sigma$. In other words, bags are finite multisets of terms. 


\paragraph{Linear extension.}
Let $\R$ be a rig (a.k.a.~semi-ring). Of particular interest are the  rigs $\B\eqdef(\set{0,1}, \max, \min)$ of booleans, $\NN \eqdef(\N, +, \times)$ of non-negative integers and $\Q\eqdef(\mathbf Q, +, \times)$ of rational numbers. We denote as $\lcomb \R\resourceTerms$ (\resp $\lcomb\R\bags$) the set of all formal (finite or infinite) linear combinations of resource terms (\resp bags) with coefficients in $\R$. If $a\in\lcomb\R\resourceTerms$ (\resp $\ms a\in\lcomb\R\bags$) and $s\in\resourceTerms$ (\resp $\ms s\in\bags$), then $a_s\in\R$ (\resp $\ms a_{\ms s}\in\R$) denotes the coefficient of $s$ in $a$ (\resp $\ms s$ in $\ms a$). As is well-known, $\lcomb\R\resourceTerms$ is endowed with a structure of $\R$-module, where addition and scalar multiplication are defined component-wise, i.e.\ for $a,b\in\lcomb\R\resourceTerms$ and $\alpha\in\R$: $(a+b)_s\eqdef a_s+b_s$, and $(\alpha a)_s\eqdef \alpha a_s$. We will write $\support a\subseteq\resourceTerms$ for the support of $a$: $\support a=\set{s\in\resourceTerms\st a_s\not=0}$.

Moreover, each resource calculus constructor extends to $\lcomb \R\resourceTerms$ component-wise, i.e.\ for any $a,a_1,\dots, a_n\in \lcomb\R\resourceTerms$ and $\ms a\in\lcomb\R\bags$, we set:
\[\textstyle
\labs x a \eqdef\sum_{s\in\resourceTerms}a_s\labs x s,\qquad
\rappl a{\ms a} \eqdef \sum_{s\in\resourceTerms,\,\ms s\in\bags} a_s\ms a_{\ms s}\rappl s{\ms s},
\]
\[\textstyle
	\mset{a_1,\dots, a_n}\eqdef\sum_{s_1,\dots,s_n\in\resourceTerms}\pa{a_1}_{s_1}\cdots \pa{a_n}_{s_n}\mset{s_1,\dots,s_n}.
\]
Notice that the last formula is coherent with the notation of the concatenation of bags as a product since it expresses a distributivity law. 
In particular, we denote $a^n\eqdef \underbrace{\mset{a,\dots, a}}_n$ and if $\Q\subseteq\R$, $\displaystyle a^!\eqdef\sum_{n\in\N}\tfrac 1{n!}a^n$.


In the case $\R=\B$, notice that $\lcomb\B\resourceTerms$ is the power-set lattice $\powerset\resourceTerms$, so that we can use the set-theoretical notation: e.g. writing $s\in a$ instead of $a_s\neq 0$, or $a\cup b$ for $a+b$. Also, in that case the preceding formulas lead to:
for $a,b\subseteq\resourceTerms$, $t\in\resourceTerms$, $\ms a\subseteq\bags$: $\labs xa\eqdef\set{\labs xs\st s\in a}$, $\prom a\eqdef\set{\mset{s_1,\cdots,s_n}\st s_1,\cdots,s_n\in a}$, $\rappl t{\ms a}\eqdef \set{\rappl t{\ms s}\st \ms s\in \ms a}$ and $\rappl a{\ms a}\eqdef \Union\set{\rappl s{\ms a}\st s\in a}$.

\paragraph{Taylor expansion.} Ehrhard and Regnier have used in~\cite{difftaylor} the rig of rational numbers to express the $\lambda$-terms as formal combinations in $\lcomb\Q\resourceTerms$. We refer to this translation $\taylorExp{(\;)}: \Lambda_+\mapsto\lcomb\Q\resourceTerms$ as the \emph{Taylor expansion} and we recall it in Table~\ref{table:taylor} by structural induction. The supports of these expansions can be seen as a map $\taylor{\;}: \Lambda_+\mapsto\lcomb\B\resourceTerms$ and directly defined by induction as in Table~\ref{table:taylor_support}. 

\begin{table}[t]
\centering
\begin{subtable}[b]{0.38\linewidth}
\begin{align*}
\taylorExp x&\eqdef x\\
\taylorExp {(\labs x M)}&\eqdef \labs x{\taylorExp M}\\
\taylorExp{(\appl MN)}&\eqdef \sum_{n\in\N}\frac{1}{n!}\rappl{\taylorExp M}{\left(\taylorExp N\right)}^n\\
\taylorExp{\left(M+N\right)}&\eqdef \taylorExp M +\taylorExp N
\end{align*}
\caption{expansion $\taylorExp{(\;)}:\ndTerms\mapsto\lcomb\Q\resourceTerms$}
\label{table:taylor}
\end{subtable}
\qquad
\begin{subtable}[b]{0.49\linewidth}
\begin{align*}
  \taylor{x}&\eqdef\set x\\
  \taylor{\lambda x.M}&\eqdef\set{\labs xs\st s\in\taylor M}\\
  \taylor{\appl MN}&\eqdef\set{\rappl s{\ms t}\st s\in\taylor M, \ms t\in\prom{\taylor N}}\\[7pt]
   \taylor{M+N}&\eqdef\taylor{M}\union\taylor{N}
\end{align*}
\caption{support $\taylor{\;}:\ndTerms\mapsto\lcomb\B\resourceTerms$}
\label{table:taylor_support}
\end{subtable}
\caption{Definition of the Taylor expansion $\taylorExp{(\;)}$ and of its support $\taylor{\;}$.}
\label{table:taylor and its support}
\end{table}

\begin{example}\label{ex:taylor_expansion}
From the above definitions we have:
\[
\taylorExp{\mathbf\Delta}=\sum_{n}\frac1{n!}\lambda x.\rappl{x}{x^n},\qquad
\taylorExp{\mathbf{\Delta_3}}=\sum_{n,m}\frac1{n!m!}\lambda x.\rappl{\rappl{x}{x^n}}{x^m}.
\]
Let us denote as $\delta_n$ (resp.\ $\delta_{n,m}$) the term $\lambda x.\rappl{x}{x^n}$ (resp.\ $\lambda x.\rappl{\rappl{x}{x^n}}{x^m}$). We can then write:
\begin{eqnarray}
	\label{eq:taylorOmega}\taylorExp{\mathbf\Omega}&=&\sum_{k}\frac1{k!}\sum_{n_0,\dots,n_k}\frac1{n_0!\cdots n_k!}\rappl{\delta_{n_0}}{\mset{\delta_{n_1},\dots,\delta_{n_k}}},\\
\label{eq:taylorOmega3}\taylorExp{\mathbf{\Omega_3}}&=&\sum_{k}\frac1{k!}\sum_{
	\substack{
		n_0,\dots,n_k\\
		m_0,\dots,m_k
	}
}\frac1{n_0!m_0!\cdots n_k!m_k!}\rappl{\delta_{n_0,m_0}}{\mset{\delta_{n_1,m_1},\dots,\delta_{n_k,m_k}}}.
\end{eqnarray}
\end{example}

 It is clear that the resource terms appearing with non-zero coefficients in $\taylorExp M$ describe the structure of $M$ taking an explicit number of times the argument of each application, and this recursively. The rôle of the rational coefficients will be clearer once defined the reduction rules over $\resourceTerms$ (see Example~\ref{ex:taylor_reduction}). 
\paragraph{Operational semantics. } Let us write $\degree xt$ for the number of free occurrences of a variable $x$ in a resource term $t$. We define the \emph{differential substitution} of a variable $x$ with a bag $\mset{s_1,\dotsc,s_n}$ in a resource term $t$, denoted $\dsubst tx{\mset{s_1,\dotsc,s_n}}$: it is a finite formal sum of resource terms, 
 which is zero whenever $\degree xt\neq n$; otherwise it is the sum of all possible terms obtained by linearly replacing each free occurrence of $x$ with exactly one $s_i$, for $i=1,\dots, n$. Formally,
\begin{equation}\label{eq:resource}
\dsubst tx{\mset{s_1,\dotsc,s_n}}\eqdef 
\begin{cases}
\displaystyle\sum_{f\in\perm n}\nsubst t{x_1}{s_{f(1)}}{x_n}{s_{f(n)}}&\text{if $\degree xt=n$,}\\
0&\text{otherwise,}
\end{cases}
\end{equation}
\noindent where $\perm n$ is the group of permutations over $n=\{1,\dots, n\}$ and $x_1,\dots, x_n$ is any enumeration of the free occurrences of $x$ in $t$, so that $\subst t{x_i}{s_{f(i)}}$ denotes the term obtained from $t$ by replacing the $i$-th free occurrence of $x$ with $s_{f(i)}$. 
Then, we give a linear extension of the differential substitution: if $a\in \lcomb\R\resourceTerms$ and $\ms a\in\lcomb\R\bags$, we set: $\dsubst ax {\ms a}= \sum_{t\in\resourceTerms, \ms s\in\bags}a_t\, \ms a_{\ms s}\dsubst tx{\ms s}$.

The resource reduction $\dbetaRed$ is then the smallest relation satisfying:
\[
\rappl{\lambda x.t}{\mset{s_1,\dotsc,s_n}}
\dbetaRed
\dsubst tx{\mset{s_1,\dotsc,s_n}}
\]
and moreover linear and compatible with the resource calculus constructors. Spelling out these two last conditions: for any $t,u\in\resourceTerms$, $\ms s\in\bags$, $a,b\in\lcomb\R\resourceTerms$, $\alpha\in\R\setminus\set 0$, whenever $t\dbetaRed a$, we have: (compatibility) $\lambda x.t\dbetaRed \lambda x.a$, $\rappl{t}{\ms s}\dbetaRed \rappl{a}{\ms s}$, $\rappl{u}{t\cdot\ms s}\dbetaRed \rappl{u}{a\cdot\ms s}$, and (linearity) $\alpha t+b\dbetaRed \alpha a+b$.


%

\begin{proposition}[\cite{difftaylor}]\label{prop:nf_resource}
Resource reduction $\dbetaRed$ is confluent over the whole $\lcomb\R\resourceTerms$ and it is normalizing on the sums in $\lcomb\R\resourceTerms$ having a finite support.
\end{proposition}
Proposition~\ref{prop:nf_resource} shows that any single resource term $t\in\resourceTerms$ has a unique normal form that we can denote as $\NF(t)$. What about possibly infinite linear combinations in $\lcomb\R\resourceTerms$? We would like to extend the normal form operator component-wise, as follows:
\begin{equation}\label{eq:normal_form}
\NF(a)\eqdef\sum_{t\in\resourceTerms} a_t\cdot\NF(t)
\end{equation}

\begin{example}\label{ex:undefined_nf}
The sum in Equation~\eqref{eq:normal_form} can be undefined for general $a\in\lcomb\R\resourceTerms$. Take $a=\lambda x.x + \rappl{\lambda x.x}{\mset{\lambda x.x}}+\rappl{\lambda x.x}{\mset{\rappl{\lambda x.x}{\mset{\lambda x.x}}}}+\dots$: any single term of this sum reduces to $\lambda x.x$, hence $\NF(a)_{\lambda x.x}$ is infinite.

Another example is given by the Taylor expansion of the term in Equation~\eqref{eq:infiniteterm}: one can check that the sum defining $\NF\bigl(\taylorExp{\pa{\appl{\mathbf\Theta}{\lambda x.(x+y)}}}\bigr)_{y}$ following Equation~\eqref{eq:normal_form} is infinite because, for all $n\in\mathbb N$, there is a resource term of the form $\rappl{t_n}{\mset{(\lambda x.x)^n,\lambda x.y}}\in\taylor{\appl{\mathbf\Theta}{\lambda x.(x+y)}}$, reducing to $y$. A closer inspection of the resulting coefficients (that we do not develop here) moreover shows that this infinite sum has unbounded partial sums in $\Q$, hence it diverges in general.

In fact, Corollary~\ref{cor:NF} ensures that if $a$ is the Taylor expansion of a strongly normalizing non-deterministic $\lambda$-term then Equation~\eqref{eq:normal_form}  is well-defined. 
\end{example}

\begin{example}\label{ex:taylor_reduction}
Recall the expansions of Example~\ref{ex:taylor_expansion} and consider $\taylorExp{(\appl{\mathbf\Delta}{\lambda x.x})}=\sum_{n,k}\frac1{n!k!}\rappl{\lambda x.\rappl{x}{\mset{x^n}}}{\mset{(\lambda x.x)^k}}$. The resource reduction applied to a term of this sum gives zero except for $k=n+1$; in this latter case we get $(n+1)!\rappl{\lambda x.x}{\mset{(\lambda x.x)^n}}$. Hence we have: $\taylorExp{(\appl{\mathbf\Delta}{\lambda x.x})}\dbetaRed\sum_{n}\frac{1}{n!}\rappl{\lambda x.x}{\mset{(\lambda x.x)^n}}$, because the coefficient $(n+1)!$ generated by the reduction step is erased by the fraction $\frac1{k!}$ in the definition of Taylor expansion. Then, the term $\rappl{\lambda x.x}{\mset{(\lambda x.x)^n}}$ reduces to zero but for $n=1$, in the latter case giving $\lambda x.x$. So we have: $\NF(\taylorExp{(\appl{\mathbf\Delta}{\lambda x.x})})=\lambda x.x=\taylorExp{(\NF(\appl{\mathbf\Delta}{\lambda x.x}))}$.

The commutation between computing normal forms and Taylor expansions has been proven in general for deterministic $\lambda$-terms \cite{bohmtaylor}\footnote{The statement proven in \cite{bohmtaylor} is actually more general, because it considers (possibly infinite) B\"ohm trees instead of the normal forms.} and witnesses the solidity of the definitions. The general case for $\ndTerms$ is still an open issue. 
\end{example}

\begin{example}\label{ex:taylor_reduction_omega}
Recall the notation of Equation~\eqref{eq:taylorOmega} from Example~\ref{ex:taylor_expansion} expressing the sum $\taylorExp{\mathbf\Omega}$. All terms with $n_0\neq k+1$ reduce to zero in one step. For $n_0=k+1$, we have that a single term rewrites to $\sum_{f\in\perm k}\rappl{\delta_{n_{f(1)}}}{\mset{\delta_{n_{f(2)}},\dots,\delta_{n_{f(k)}}}}$, which is a sum of terms still in $\taylor{\mathbf\Omega}$, but with smaller size \wrt the redex. Therefore, every term of $\taylorExp{\mathbf\Omega}$ eventually reduces to zero, after a reduction sequence whose length depends on the initial size of the term, and whose elements are sums with supports in $\taylor{\mathbf\Omega}$. 
This is in some sense the way Taylor expansion models the unbounded resource consumption of the loop $\mathbf\Omega\betaRed\mathbf\Omega$ in $\lambda$-calculus. 

We postpone the discussion of the reduction of $\taylorExp{\mathbf\Omega_3}$ until Remark~\ref{rk:counterexample_singletons}.
\end{example}

\subsection{Finiteness structures induced by antireduction}

Let us get back to Equation~\eqref{eq:normal_form}, and consider it pointwise:
for all $s\in\resourceTerms$ in normal form, we want to set 
$\NF(a)_s = \sum_{t\in\resourceTerms} a_t\cdot\NF(t)_s$.
Notice that this series can be obtained as the inner product
between $a$ and the vector $\cone s$ with
$(\cone s)_t=\NF(t)_s$: one can think of $s$ as 
a test, that $a$ passes whenever the sum converges.

There is one very simple condition that one can impose on a formal series to
ensure its convergence: just assume there are finitely many non-zero terms.
This seemingly dull remark is in fact the starting point of the definition 
of finiteness spaces, introduced by Ehrhard \cite{ehrhardfs} and discussed in the Introduction. 
The basic construction is that of a finiteness structure:
\begin{definition}[\cite{ehrhardfs}]\label{def:finiteness}
	Let $A$ be a fixed set.
	A \emph{structure} on $A$ is any set of subsets $\fA\subseteq\powerset A$.
	For all subsets $a$ and $a'\subseteq A$, we write $a\perp a'$
	whenever $a\inter a'$ is finite.
	For all structure $\fA\subseteq \powerset A$, we
	define its dual $\dual\fA=\set{a\subseteq A\st \forall a'\in\fA,\ a\perp a'}$.
	A \emph{finiteness structure} on $A$ is any such $\dual\fA$.
\end{definition}
Notice that:  $\fA\subseteq\fA^{\perp\perp}$, also $\fA\subseteq\fA'$ entails $\dual{\fA'}\subseteq\dual\fA$, hence $\dual{\fA}=\fA^{\perp\perp\perp}$.

Let $\fC_0=\set{\support{\cone s}\st s\in\resourceTerms\text{, in normal form}}\subseteq\powerset{\resourceTerms}$, we obtain that $\support a\in\dual{\fC_0}$ iff Equation~\eqref{eq:normal_form} involves only pointwise finite sums.
So, one is led to focus on support sets only, leaving out
coefficients entirely. Henceforth, unless specified otherwise, we will thus
stick to the case of $\R=\B$, and use set-theoretical notations only.
This approach of ensuring the normalization of Taylor expansion \emph{via} a finiteness structure was first used by Ehrhard \cite{Ehrhard10lics} for a non-deterministic variant of System $\cF$.
Our paper strengthens Ehrhard's result in several directions. In order to state them, we introduce a construction of finiteness structures on $\resourceTerms$ induced by anticones for the reduction order defined as follows:
\begin{definition}\label{def:reduce}
	For all $s,t\in\resourceTerms$, we write $t\reducesTo s$ whenever there
	exists a reduction $t\dbetaRed^* a$ with $s\in a$.
\end{definition}
It should be clear that this defines a partial order relation
(in particular we have antisymmetry because $\dbetaRed$ terminates).
\begin{definition}\label{def:cones}
  If $a\subseteq\resourceTerms$, $\cone a \eqdef \set{t\in
    \resourceTerms\st \exists s\in a,\ t\ge s}$ is the cone of
	antireduction over $a$.\footnote{Observe that $\cone{\set s}=\cone s$ (up to the identification of $\lcomb\B\resourceTerms$ with $\powerset\resourceTerms$).} For any structure
  $\fT\subseteq\powerset{\resourceTerms}$, we write
  $\cones\fT=\set{\cone a\st a\in\fT}$.
\end{definition}
We can consider the elements of $\fT$ as tests, and say a set $a\subseteq\resourceTerms$
passes a test $a'\in\fT$ iff $a\perp\cone a'$. The structure of sets 
that pass all tests is exactly $\dual{\cones{\fT}}$.
Then Ehrhard’s result can be rephrased as follows:
\begin{theorem}[\cite{Ehrhard10lics}]\label{th:thomas_theorem}
	If $M\in\ndTerms$ is typable in System $\cF$ then 
	$\taylor M\in\dual{\cones{\fS_{\text{sgl}}}}$
	where $\fS_{\text{sgl}}\eqdef\set{\set s\st s\in\resourceTerms}$.
\end{theorem}
Notice that, in contrast with the definition of $\fC_0$, $\fS_{\text{sgl}}$ is in fact not restricted to normal forms.
Our paper extends this theorem in three directions: first, one can relax the
condition that $M$ is typed in System $\mathcal F$ and require only that $M$ is strongly
normalizable; second, the same result can be established for sets of “tests”
larger (hence more demanding) than $\fS_{\text{sgl}}$; third, the implication
can be reversed for a suitable set of tests $\fT$, \ie $M$ is strongly
normalizable iff $\taylor M\in\dual{\cones\fT}$ (and we do need $\fT$ to 
provide  more tests than just singletons: see Remark~\ref{rk:counterexample_singletons}).
In order to state our results precisely, we need to introduce a few more 
notions.

\begin{definition}\label{def:finiteness structures}
We say that a resource term $s$ is \emph{linear} whenever each bag appearing in $s$ has cardinality $1$. A set $a$ of resource terms is said \emph{linear} whenever all its elements are \emph{linear}. We say $a$ is \emph{bounded} whenever there exists a number $n\in\N$ bounding the cardinality of all bags in all terms in $a$. 
We then write
\begin{align*}
\linearStructure&
\eqdef\set{a\subseteq\resourceTerms\st \text{$a$ linear}}&
\text{and}&&
\boundedStructure&
\eqdef\set{a\subseteq\resourceTerms\st\text{$a$ bounded}}.
\end{align*}
We also denote as $\lEmbed M$ the subset of the linear resource terms in $\taylor M$. Notice that 
 $\lEmbed{M}$ is always non-empty and can be directly defined by replacing the definition of $\taylor{\appl MN}$ in Table~\ref{table:taylor_support} with: $
 \lEmbed{\appl MN}\eqdef\set{\rappl s{\mset t}\st s\in\lEmbed M, t\in\lEmbed N}$.
\end{definition}
Observe that $\fS_{\text{sgl}},\linearStructure\subseteq\boundedStructure$.

We can sum up our results Corollary~\ref{cor:realizability} and Theorem~\ref{theorem:finitary:implies:SN} as follows:
\begin{theorem}\label{th:finiteness=termination}
	Let $M\in\ndTerms$:
	\begin{itemize}
		\item 
			If $M$ is strongly normalizing, then $\taylor
			M\in\dual{\cones\fT}$ as soon as $\fT\subseteq\boundedStructure$.
		\item 
			If $\taylor M\in\dual{\cones\fT}$ with
			$\linearStructure\subseteq\fT$, then $M$ is strongly normalizing.
	\end{itemize}
\end{theorem}

%% file: C-systemD.tex


\begin{table}[t]
{\small
\begin{subtable}[t]{\linewidth}
	\[
		A,B\eqdef X \mid A\to B\mid A\cap B
	\]
	\caption{Grammar of types, $X$ varying over a denumerable set of propositional variables}
\end{subtable}

\hrulefill

\medskip
\begin{subtable}[t]{\linewidth}
	\centering 
	\begin{prooftree}
		\Hypo{\text{for } i\in\{1,2\}}
		\Infer{1}{A_1\inter A_2\subtype A_i}
	\end{prooftree}
	\qquad
	\begin{prooftree}
		\Hypo{A'\subtype A}
		\Hypo{B\subtype B'}
	 	 \Infer{2}{A\to B\subtype A'\to B'}
	\end{prooftree}
	\qquad
	\begin{prooftree}
		\Infer{0}{(A\to B)\inter(A\to C)\subtype A\to (B\inter C)}
	\end{prooftree}
	\medskip
	\caption{Rules generating the subtyping relation $\subtype$}\label{table:rules_subtyping}
\end{subtable}

\hrulefill

\medskip
\begin{subtable}[t]{\linewidth}
\centering
	\setkeys{ebproof}{center=false}
	\begin{prooftree}
		\VarRule{\Gamma,x:A\vdash x:A}
	\end{prooftree}
	\qquad
	\begin{prooftree}
		\Hypo{\Gamma,x:A\vdash M:B}
		\AbsRule{\Gamma\vdash \labs xM:A\to B}
	\end{prooftree}

	\medskip	
	\begin{prooftree}
		\Hypo{\Gamma\vdash M:A\to B}
		\Hypo{\Gamma\vdash N:A}
		\AppRule{\Gamma\vdash \appl MN: B}
	\end{prooftree}
	\qquad
	\begin{prooftree}
		\Hypo{\Gamma\vdash M:A}
		\Hypo{\Gamma\vdash M:B}
		\IntRule{\Gamma\vdash M: A\inter B}
	\end{prooftree}
	
	\medskip
	\begin{prooftree}
		\Hypo{\Gamma\vdash M:A}
		\Hypo{\Gamma\vdash N:A}
		\PlusRule{\Gamma\vdash M+N:A}
	\end{prooftree}
	\qquad
	\begin{prooftree}
		\Hypo{\Gamma\vdash M:A}
		\Hypo{A\subtype A'}
		\SubRule{\Gamma\vdash M:A'}
	\end{prooftree}

	\medskip
	\caption{Derivation rules}
\end{subtable}
}
\caption{the intersection type assignment system $\mathcal D_+$ for $\ndTerms$}\label{fig:D+}
\end{table}
Intersection types are well known, as well as their relation with normalizability. We refer to~\cite{interbcd} for the original system with subtyping characterizing the set of strongly normalizing $\lambda$-terms, and~\cite{Bakel92} and~\cite{lcalculen} for simpler systems. However, as far as we know, the literature about intersection types for  non-deterministic $\lambda$-calculi is less well established and in fact we could find no previous characterization of strong normalization in a non-deterministic setting. Hence, we give in Table~\ref{fig:D+} a variant of Krivine's system $\mathcal D$~\cite{lcalculen}, characterizing the set of strongly normalizing terms in $\ndTerms$. In this section, we only prove that strongly normalizing terms are typable (Theorem~\ref{theorem:sn:typable}): the reverse implication follows from the rest of the paper (see Figure~\ref{fig:main_results}). These techniques are standard.
%
\begin{remark}
Krivine's original System $\mathcal D$ does not have $\SubRuleName$ and $\PlusRuleName$, but it has the two usual elimination rules for intersection (here derivable). 
The rule $\PlusRuleName$ is necessary to account for non-determinism, however adding just $\PlusRuleName$ to System $\mathcal D$ is misbehaving. We can find terms $M$ and $N$, and a context $\Gamma$ such that $\appl MN$ is typable
in System $\mathcal D$ with $\PlusRuleName$ under the context $\Gamma$ but $M$ is not: take 
$\Gamma=x:A\to B\inter B',y:A\to B\inter B'',z:A$,
observe that $\appl{x+y} z=\appl xz+\appl yz$,
and thus $\Gamma\vdash \appl {x+y}z:B$ but
$x+y$ is not typable in $\Gamma$.
This is the reason why we introduce subtyping.
\end{remark}

\begin{theorem}
	\label{theorem:sn:typable}
	For all $M\in\ndTerms$, if $M$ is strongly normalizable, then 
	there exists a derivable judgement $\Gamma\vdash M:A$ in system $\dplus$.
\end{theorem}
\begin{proof}[Sketch]
	For $M$ a strongly normalizable term, let
	$\redlen M$ be the maximum length of a reduction from $M$, and $\size M$ the number of symbols occurring in $M$.
	By well-founded induction on the pair $(\redlen M,\size M)$ we prove that
	there exists $n_M\in\N$ such that for all type $B$ and all $n\ge n_M$, there is a context $\Gamma$ and a sequence
	$(A_1,\dotsc,A_n)$ of types such that $\Gamma\vdash M:A_1\to\cdots\to A_n\to B$.
	
	The proof splits depending on the structure of $M$. In case $M=M_1+M_2$, we apply the induction hypothesis on both $M_1$ and $M_2$ and conclude by rule $\SubRuleName$ and a contravariance property: $\Gamma\vdash M:A$ whenever $\Gamma'\vdash M:A$ and $\Gamma\subtype\Gamma'$.



In case of head-redexes, i.e.\ $M=\appl {\appl{\labs x N}P} {M_1\cdots M_q}$,  we apply the induction hypothesis on $M'=\appl{\subst NxP}{M_1\cdots M_q}$ and on $P$. Then, we conclude via a subject expansion lemma stating that: $\Gamma\vdash \appl{\labs x N}{P\,M_1\cdots M_n}:A$, whenever $\Gamma\vdash \appl{\subst NxP}{M_1\cdots M_n}:A$ and there exists $B$ such that $\Gamma\vdash P:B$. 

The other cases are similar to the first one. 
\qed
\end{proof}

%% file: D-sn_fin.tex
%

This section proves Corollary~\ref{cor:realizability}, giving sufficient conditions (to be dispersed, hereditary and expandable, see  resp.\ Definition~\ref{def:dispersion}, \ref{def:hereditary} and \ref{def:expandable}) over a structure $\fT$ in order to have all cones $\cone a$ for $a\in\fT$ dual to the Taylor expansion of any strongly normalizing non-deterministic $\lambda$-term. 

It is easy to check that these conditions are satisfied by the structures $\boundedStructure$ of bounded sets and $\linearStructure$ of sets of linear terms (Definition~\ref{def:finiteness structures}). Moreover, as an immediate corollary one gets also that any subset $\fT\subseteq\boundedStructure$ is also such that $\taylor{M}\in\dual{\cones{\fT}}$ for any strongly normalizable term $M\in\ndTerms$, so getting the first Item of Theorem~\ref{th:finiteness=termination}.


Thanks to previous Theorem~\ref{theorem:sn:typable} we can prove Corollary~\ref{cor:realizability} by a realizability technique on the intersection type system $\mathcal D_+$. For a fixed structure $\fS$, we associate  with any type $A$ a realizer $\realize[\fS] A$ (Definition~\ref{def:realizers}). In the case $\fS$ is adapted (Definition~\ref{def:adapted}), we can prove that $\realize[\fS] A$ contains the Taylor expansion of any term of type $A$ and that it is contained in $\fS$ (Theorem~\ref{theorem:realizability}). These definitions and theorem are adapted from Krivine's proof for System $\mathcal D$~\cite{lcalculen}. 

The crucial point is then to find structures $\fS$ which are \emph{adapted}: here is our contribution. The structures that we study have the shape $\dual{\cones{\fT}}$, so that we are speaking of the interaction with cones of anti-reducts of tests in a structure $\fT$. Intuitively, $\fT$ is a set of tests that can be passed by any term typable in System $\mathcal D_+$ (hence by any strongly normalizing term). We prove (Lemma~\ref{lemma:condition_adaptedness}) that for a structure $\fT$,  being dispersed, hereditary and expandable is sufficient to guarantee that the dual structure $\dual{\cones\fT}$ is adapted, then achieving Corollary~\ref{cor:realizability}.

\begin{definition}[Functional]\label{def:functional}
Given two structures $\fS,\fS'\subseteq\powerset{\resourceTerms}$, we define the structure $\fS\to\fS'\eqdef\set{f\subseteq\resourceTerms\st \forall a\in\fS, \rappl f{\prom a}\in\fS'}$.
\end{definition}

\begin{definition}[Saturation]\label{def:saturation}
	Let $\fS, \fS'\subseteq\powerset{\resourceTerms}$. We say $\fS'$ is \emph{$\fS$-saturated} if $\forall e,f_0,\dots,f_n\in\fS$, $\rappl{\dsubst{e}x{\prom{f_0}}}{\prom{f_1}\dots\prom{f_n}}\in\fS'$ implies  $\rappl{\lambda x.e}{\prom{f_0}\,\prom{f_1}\dots\prom{f_n}}\in\fS'$.
\end{definition}

\begin{definition}[Adaptedness]\label{def:adapted} 
For all $\fS\subseteq\powerset{\resourceTerms}$ we set $\neutralSets{\fS} \eqdef \set{\set{x}\st x\in\variables}\mathbin\cup\set{\rappl{x} \prom{a_1}\cdots\prom{a_n}\st x\in\variables,\ n>0 \text{ and } \forall i,\ a_i\in\fS}$
and say $\fS$ is \emph{adapted} if:
\begin{enumerate}
\item $\fS$ is $\fS$-saturated;\label{cond:adaptedness_saturate}
\item $\neutralSets\fS\subset (\fS\to\neutralSets\fS)\subset(\neutralSets\fS\to\fS)\subset\fS$;\label{cond:adaptedness_inclusions}
\item $\fS$ is closed under finite unions: $\forall b,b'\in \fS$, $b\cup b'\in\fS$.\label{cond:adaptedness_unions}
\end{enumerate}
\end{definition}

\begin{definition}[Realizers]\label{def:realizers}
Let $\fS\subseteq\powerset{\resourceTerms}$. To each type $A$ of System $\mathcal D_+$, we associate a structure
$\realize[\fS] A$ defined inductively ($X$ being a propositional variable):
\begin{align*}
	\realize[\fS] X&\eqdef\fS,&
	\realize[\fS] {A\to B}&\eqdef\realize[\fS]{A}\to\realize[\fS]{B},&
	\realize[\fS] {A\inter B}&\eqdef\realize[\fS]{A}\cap\realize[\fS]{B}.
\end{align*}
\end{definition}
\begin{lemma}\label{lemma:relizers_saturated}
Let $\fS$ be an adapted structure, then for every type $A$, $\realize[\fS] A$ is $\fS$-saturated, closed under finite unions and $\neutralSets\fS\subseteq\realize[\fS] A\subseteq\fS$.
\end{lemma}
\begin{lemma}[Adequacy]\label{lemma:realizability}
If $\fS\subseteq\powerset{\resourceTerms}$ is adapted, $x_1:A_1,\dotsc,x_n:A_n\vdash M:B$ and for all $1\le i\le n$, $a_i\in\realize[\fS]{A_i}$, then
	$\dsubst{\taylor M}{x_1,\dotsc,x_n}{\prom{a_1},\dotsc,\prom{a_n}}\in\realize[\fS]{B}$.
\end{lemma}
\begin{proof}[Sketch]
By structural induction on the derivation of $x_1:A_1,\dotsc,x_n:A_n\vdash M:B$. The cases where the last rule is \AbsRuleName{} or \PlusRuleName{} use respectively the facts that the realizers are saturated and closed by finite unions. The case where the last rule is \SubRuleName{} is an immediate consequence of the induction hypothesis and of a lemma stating that $A\subtype B$ implies $\realize[\fS] A\subseteq\realize[\fS] B$.\qed
\end{proof}

\begin{theorem}\label{theorem:realizability}
If  $\fS\subseteq\powerset{\resourceTerms}$ is adapted and $M$ is typable in System $\mathcal D_+$, then $\taylor{M}\in\fS$. 
\end{theorem}
\begin{proof}
Let $\Gamma\vdash M:B$. For any $x:A$ in $\Gamma$, $\{x\}\in\neutralSets{\fS}\subset\realize[\fS]{A}$ by Lemma~\ref{lemma:relizers_saturated} and Definition~\ref{def:adapted}. Hence, $\taylor{M}=\dsubst{\taylor M}{x_1,\dotsc,x_n}{\prom{\set{x_1}},\dotsc,\prom{\set{x_n}}}\in\realize[\fS]{B}$ by Lemma~\ref{lemma:realizability}. Again by Lemma~\ref{lemma:relizers_saturated}, $\realize[\fS]{B}\subseteq\fS$, so $\taylor{M}\in\fS$.\qed
\end{proof}

Now we look for conditions to ensure that a structure $\fS$ is \emph{adapted}. These conditions (Definition~\ref{def:dispersion}, \ref{def:hereditary} and~\ref{def:expandable}) are quite technical  but they are easy to check, in particular the structures $\boundedStructure$ and $\linearStructure$ enjoy them (Remark~\ref{rk:examples_realizabilit_conditions}). 


\begin{definition}
The \emph{height} $\height{s}$ of a resource term $s$ is defined inductively: $\height{x}\eqdef 1$,  $\height{\lambda x.s}\eqdef 1+\height{s}$ and $\height{\rappl{s_0}{\mset{s_1,\dotsc,s_n}}}\eqdef 1+\max_i(\height{s_i})$.
\end{definition}

\begin{definition}[Dispersed]\label{def:dispersion}
A set $a\subseteq\resourceTerms$ is \emph{dispersed} whenever for all $n\in\N$, 
and all finite set $V$ of variables, the set $\set{s\in a\st \height s\le n\text{ and }\fv s\subseteq V}$ is finite. A structure $\fS$ is dispersed whenever $\forall a\in \fS$, $a$ is dispersed. 
\end{definition}
%

\begin{definition}\label{def:project_set}
	Let $s\in\resourceTerms$ and $x\not\in\fv s$.
	We define \emph{immediate subterm projections}
	$\proj xs\in\powerset\resourceTerms$,
	$\proj 0s\in\powerset\resourceTerms$,
	$\mproj 1s\in\powerset{\prom\resourceTerms}$ and 
	$\proj 1s\in\powerset{\resourceTerms}$ as follows:
	\begin{itemize}
		\item if $s=λx.t$ then $\proj xs=\set t$; otherwise
			$\proj xs=\emptyset$;
		\item if $s=\rappl {t}{\ms u}$ then $\proj 0s=\set t$,
			$\mproj 1s=\set{\ms u}$, $\proj 1s=\support{\ms u}$;
			otherwise $\proj 0s=\mproj 1s=\proj 1s=\emptyset$.
	\end{itemize}
\end{definition}
Observe that up to $\alpha$-conversion and the hypothesis $x\not\in\fv s$,
the abstraction case is exhaustive. These functions obviously extend to sets of terms, up to some 
care about free variables. 
If $V\subseteq\variables$ is a set of variables, we write  $\resourceTerms[V]$ for the set
of  resource $\lambda$-terms with free variables in $V$.
\begin{definition}\label{def:project_structure}
	For all $V\subseteq \variables$ and $a\subseteq\resourceTerms[V]$,
	let
	\begin{align*}
		\proj 0a &\eqdef \Union_{s\in a}\proj 0s\subseteq\resourceTerms[V],&
		\mproj 1a &\eqdef \Union_{s\in a}\mproj 1s\subseteq\prom{\resourceTerms[V]},&
		\proj 1a &\eqdef \Union_{s\in a}\proj 1s\subseteq\resourceTerms[V],
	\end{align*}
	and if moreover $x\not\in V$, then let $\proj xa \eqdef \Union_{s\in a}\proj xs\subseteq\resourceTerms[V\union\set x]$.
\end{definition}

\begin{definition}[Hereditary]\label{def:hereditary}
A structure $\fS\subseteq\powerset{\resourceTerms}$ is said to be \emph{hereditary} if, $\fS$ is downwards closed, and for all $a\in\fS$, $\proj 0a\in\fS$, $\proj 1a\in\fS$ and for all $x\in\variables\setminus\fv a$, $\proj xa\in\fS$.
\end{definition}

\begin{definition}[Expandable]\label{def:expandable}
	A structure  $\fS\subseteq\powerset{\resourceTerms}$ is said to be \emph{expandable} 
	if, for all $x\in\variables$ and all $a\in\fS$, we have
	$\set{\rappl s{\mset x}\st s\in a}\in\fS$.
\end{definition}

\begin{remark}\label{rk:examples_realizabilit_conditions}
The structures $\fS_{\text{sgl}}$ (Theorem~\ref{th:thomas_theorem}) and $\linearStructure$, $\boundedStructure$ (Definition~\ref{def:finiteness structures}) are dispersed and expandable. The last two are also hereditary, while $\fS_{\text{sgl}}$ is not. 
\end{remark}
\begin{lemma}\label{lemma:condition_adaptedness}
	For any structure $\fT$ which is dispersed, hereditary and expandable, we have that $\finitarySets[\fT]$ is adapted.
\end{lemma}
\begin{proof}[Sketch]
One has to prove the three conditions of Definition~\ref{def:adapted}. Heredity and dispersion are used to obtain saturation  (Condition~\ref{cond:adaptedness_saturate}). The inclusions of Condition~\ref{cond:adaptedness_inclusions} need all hypotheses and an auxiliary lemma proving that $\neutralSets{\finitarySets[\fT]}\subseteq\finitarySets[\fT]$, for which dispersion is crucial. Finally, the closure under finite unions (Condition~\ref{cond:adaptedness_unions}) is satisfied by all finiteness structures. 
\qed
\end{proof}
\begin{corollary}\label{cor:realizability}
Let $\fT\subseteq\powerset\resourceTerms$ be dispersed, hereditary and expandable. For every strongly normalizable term $M$, we have $\taylor{M}\in\dual{\cones{\fT}}$. Hence, this holds for $\fT\in\set{\linearStructure, \boundedStructure}$ and for any of their subsets, such as $\fS_{\text{sgl}}\subset\boundedStructure$.
\end{corollary}
\begin{proof}
The general statement follows from Theorems~\ref{theorem:sn:typable},~\ref{theorem:realizability} and Lemma~\ref{lemma:condition_adaptedness}. Remark~\ref{rk:examples_realizabilit_conditions} implies $\taylor{M}\in\dual{\cones{\fT}}$ for $\fT\in\set{\boundedStructure, \linearStructure}$. The rest of the statement follows because, for any structures $\fS,\fS'$, $\fS\subseteq\fS'$ implies $\dual{\fS'}\subseteq\dual{\fS}$.\qed
\end{proof}

%% file: E-fin_sn.tex

In this section we prove Theorem~\ref{theorem:finitary:implies:SN}, giving a sufficient condition for a structure $\fT$ to be able to test strong normalization. The condition is that $\fT$  includes at least $\linearStructure$, i.e.\ the set of all sets of linear terms\footnote{This condition can be slightly weakened replacing $\linearStructure$ with: $\set{a\subseteq\resourceTerms\st \text{$a$ linear and $\fv a$ finite}}$. 
However, we prefer to stick to the more intuitive definition of $\linearStructure$.}.

The proof is by contraposition, suppose that $M$ is divergent, then $\taylor{M}$ is not dual to some cone $\cone{a}$, with $a\in\linearStructure$. The proof enlightens two kinds of divergence in $\lambda$-calculus: the one generated by looping terms: $\mathbf\Omega\betaRed\mathbf\Omega\betaRed\dots$ and the other generated by infinite reduction sequences $(M_i)_{i\in\N}$ with an infinite number of different terms: $\mathbf{\Omega_3}\betaRed\appl{\mathbf{\Omega_3}}{\mathbf{\Delta_3}}\betaRed\appl{\appl{\mathbf{\Omega_3}}{\mathbf{\Delta_3}}}{\mathbf{\Delta_3}}\betaRed\dots$ (see Example~\ref{ex:lambda-terms}). 

In the first case, the cone $\cone{\ell(\mathbf\Omega)}$ of the linear expansion (Definition~\ref{def:finiteness structures}) of the looping term $\mathbf\Omega$ suffices to show up the divergence, since $\taylor{\mathbf\Omega}\inter\cone{\ell(\mathbf\Omega)}$ is infinite. Indeed the Taylor expansion of a looping term, say $\taylor{\mathbf\Omega}$, is a kind of ``contractible space'', where any resource term reduces to a smaller term within the same Taylor expansion or vanishes (see Example~\ref{ex:taylor_reduction_omega}). In particular, there are unboundedly large resource terms reducing to the linear expansion $\ell(\mathbf\Omega)$. 

In the case of an infinite reduction sequence of different terms, one should take, basically, the cone of all linear expansions of the terms occurring in the sequence: the linear expansion of a single term (or of a finite set of terms) might not suffice to test this kind of divergence. 
For example, $\taylor{\mathbf{\Omega_3}}\inter\cone{\ell(\mathbf{\Omega_3})}$ is finite, while $\taylor{\mathbf{\Omega_3}}\inter\cone{\{\ell(\mathbf{\Omega_3}), \ell(\appl{\mathbf{\Omega_3}}{\mathbf{\Delta_3}}), \dots
\}}$ is infinite, so $\taylor{\mathbf{\Omega_3}}\!\notin\!\dual{\cones{\linearStructure}}.$ 

In the presence of the non-deterministic sum $+$, we have a third kind of divergence, which is given by infinite reduction sequences of terms $(M_i)_{i\in\N}$ which are pairwise different but whose Taylor expansion support repeats infinitely many times: 
consider, \eg, the reducts of $\appl{\mathbf\Theta}{(\lambda x.x+y)}$. 
We prove that this kind of divergence is much more similar to a loop rather than to a sequence of different $\lambda$-terms. In particular, there is a single linear resource term (depending on the reduction sequence) whose cone is able to show up the divergence. Indeed, most of the effort in the proof of Theorem~\ref{theorem:finitary:implies:SN} is devoted to deal with this kind of ``looping Taylor expansion''. Namely, Definition~\ref{def:partial_red} gives a notion of non-deterministic reduction $\partialBetaRed$ allowing Lemma~\ref{lemma:infinite:bounded:reduction}, which is the key statement used in the proof of Theorem~\ref{theorem:finitary:implies:SN} for dealing with both the divergence of looping terms (like $\mathbf\Omega$) and that of looping Taylor expansions (like $\appl{\mathbf\Theta}{(\labs x{x+y})}$). 


\medskip

We introduce a reduction rule
$\partialBetaRed$ on $\ndTerms$ which corresponds to one step of
$\beta$-reduction and a potential loss of some addenda in a term. For
that, we need an order $\ndge$ on $\ndTerms$ expressing this
loss. For instance, $\appl{\mathbf\Theta}{(\labs x{x+y})}+y\ndge \appl{\mathbf\Theta}{\labs x{x}}$, thus $\appl{\mathbf\Theta}{(\labs x{x+y})}\partialBetaRed^\ast \appl{\mathbf\Theta}{\labs xx}$, and similarly, $\appl{\mathbf\Theta}{(\labs x{x+y})}\partialBetaRed^\ast y$.
 
\begin{definition}[Partial reduction]\label{def:partial_red}
  We write $M\partialBetaRed N$ if there exists $P$ such that
  $M\betaRed P$ and $P\ndge N$, where the partial order $\ndle$ over
  $\Lambda_+$ is defined as the least order such that 
    $M\ndle M+N$; $N\ndle M+N$ 
  and if $M\ndle N$ then: $M+P\ndle N+P$, $\labs x M\ndle \labs xN$, $\appl M P\ndle
  \appl NP$, and $\appl PM\ndle \appl PN$.

A reduction $M\partialBetaRed N$ is at \emph{top
    level} if $M=\appl{\labs x{M'}}{M''}\rightarrow
  M'[M''/x]\ndge N$.
\end{definition}
Write $s\strictReducesTo t$ whenever $s\reducesTo t$ (Definition~\ref{def:reduce}) and $s\not=t$:
this is a strict partial order relation. 
\begin{lemma}
  \label{lemma:resource:antireduction} \label{lemma:resource:antireduction:toplevel}

  If $M\partialBetaRed N$ and $t\in\taylor N$, then there exists
  $s\in\taylor M$ such that $s\reducesTo t$.
  If moreover, $M\partialBetaRed N$ is at top level, then $s\strictReducesTo t$.
\end{lemma}
  


\begin{lemma}
	\label{lemma:partial:reduction:preserves:finiteness}
	Let $M\partialBetaRed N$ and $u\in\resourceTerms$. If
	$\taylor N\cap\cone u$ is infinite, then  $\taylor M\cap\cone u$ is also infinite.
\end{lemma}

\begin{definition}
The \emph{height} $\height{M}$ of a term $M\in\ndTerms$ is defined inductively as follows:
$\height{x}\eqdef 1$,  $\height{\labs xM}\eqdef 1+\height{M}$, $\height{\appl{M}{N}}\eqdef 1+\max(\height M,\height N)$ and $\height{{M}+{N}}\eqdef \max(\height M,\height N)$.
\end{definition}



\begin{lemma}
	\label{lemma:infinite:bounded:reduction}
	Let $(M_i)_{i\in\N}$ be a sequence. If $\forall i\in\N$, $M_i\partialBetaRed
	M_{i+1}$ and 
        $(\height{M_i})_{i\in\N}$ is bounded, then
        there exists a linear term $t$ such that 
	$\taylor{M_0}\cap\cone t$ is infinite. 
\end{lemma}
\begin{proof}[Sketch]
  First, it is sufficient to address the case of a sequence $(S_i)_{i\in\N}$ of simple
  terms (i.e.\ without $+$ as the top-level constructor) such that $S_i\partialBetaRed S_{i+1}$ for
  all $i\in\N$ and $(\height{S_i})_{i\in\N}$ is bounded. Besides, by
  Lemma~\ref{lemma:partial:reduction:preserves:finiteness}, it is
  sufficient to have $\taylor{S_{i_0}}\cap\cone t$ infinite for some
  $i_0\in\N$.

  Then, by induction on $h=\max\set{\height{S_i}\st i\in\N}$, we show
  that there exists $i_0\in\N$ and a sequence
  $(s_j)_{j\in\N}\in\taylor{S_{i_0}}^\N$ such that $s_0$ is linear
  and, for all $j\in\N$, $s_{j+1}\strictReducesTo s_j$. Since
  $\strictReducesTo$ is a strict order relation, this implies that the
  set $\set{s_j\st j\in\N}\subseteq\taylor{S_{i_0}}\inter\cone{s_0}$
  is infinite.



  First assume that there are infinitely many top level reductions.
  Observe that, since $\height{S_i}\le h$ and
  $\fv{S_i}\subseteq\fv{S_0}$ for all $i$, the set
  $\set{\taylor{S_i}\st i\in\N}$ is finite.  Hence there exists an
  index $i_0\in\N$ such that $\set{i\in\N\st
    \taylor{S_i}=\taylor{S_{i_0}}}$ is infinite.  As there are
  infinitely many top level reductions, there are $i_1$ and $i_2$
  such that $i_0<i_1<i_2$, the reduction $i_1$ is at top level and
  $\taylor{S_{i_2}}=\taylor{S_{i_0}}$. We inductively define the required sequence
  by choosing arbitrary
  $s_0\in\lEmbed{S_{i_0}}\subseteq\taylor{S_{i_0}}$, and by iterating
  Lemma~\ref{lemma:resource:antireduction}: for
  $s_j\in\taylor{S_{i_0}}=\taylor{S_{i_2}}$,  we obtain $s_{j+1}\in \taylor{S_{i_0}}$ with $s_{j+1}
  \strictReducesTo s_j$  since the reduction $i_1$
  is at top level.

  Now assume that there are only finitely many top level reductions.
  Let $i_1$ be such that no reduction $S_i\partialBetaRed S_{i+1}$
  with $i\ge i_1$ is at top level.
  Either for all $i\ge i_1$, $S_i=\labs x{M'_i}$ with
  $M'_i\partialBetaRed M'_{i+1}$, and we conclude by applying the
  induction hypothesis to the sequence $(M'_{i+i_1})_{i\in\N}$;
  or for all $i\ge i_1$, $S_i=\appl{M'_i}{N'_i}$ so that $(M'_i)_{i\ge
    i_1}$ and $(N'_i)_{i\ge i_1}$ are sequences of terms, 
	at least one of them involving infinitely many partial reductions.
  In this case, assume for instance that we can extract from $(M'_i)_{i\ge i_1}$
	an infinite subsequence $(M'_{\phi(i)})_{i\in\N}$ of partial reductions.
	It provides $i'_0$ and a sequence $(s'_j)\in
  	\mathcal T(M'_{\phi(i'_0)})^N$
  with $s'_0$ linear and, 
$s'_{j+1}\strictReducesTo s'_j$. Fix
  $t\in\ell(N'_{\phi(i'_0)})$ arbitrarily. So we set
  $i_0=\phi(i'_0)$ and $s_j=\rappl{s'_j}{\mset t}$ for all $j\in\N$.\qed
\end{proof}


\begin{theorem}
  \label{theorem:finitary:implies:SN}
  Let $\fT$ be a structure such that
  $\linearStructure\subseteq\fT$.  If $\taylor
  M\in\finitarySets[\fT]$ then $M$ is strongly normalizable. In
  particular, this holds for $\fT\in\set{\linearStructure,
    \boundedStructure}$.
\end{theorem}
\begin{proof}
  Assume that $(M_i)_{i\in\N}$ is such that $M=M_0$ and for all $i$,
  $M_i\betaRed M_{i+1}$.  We prove that $\taylor
  M\not\in\finitarySets[\fT]$ by exhibiting $a\in\fT$ such that
  $\taylor M\not\perp \cone a$.

  If $(\height{M_i})_{i\in\N}$ is bounded, then fix $a=\set t$ with
  $t$ given by Lemma~\ref{lemma:infinite:bounded:reduction}.
	
  Otherwise, $\forall i\in\N$, fix $t_i\in\lEmbed{M_i}$ such that
  $\height{t_i}=\height{M_i}$.
  Lemma~\ref{lemma:resource:antireduction} implies that there is
  $s_i\in\taylor M$ such that $s_i\reducesTo t_i$.  Denote by $\size{s}$ the number of symbols occurring in $s$. Since there is no duplication in reduction $\dbetaRed$ it should be clear that
if $s\reducesTo t$ then $\size s\ge\size t$. Besides, $\size s\ge\height s$. Therefore, since
  $\set{\height{M_i}\st i\in\N}$ is unbounded, $\set{\height{t_i}\st
    i\in\N}$, $\set{\size{t_i}\st i\in\N}$ and $\set{\size{s_i}\st
    i\in\N}$ are unbounded. Fix $a=\Union_{i\in\N}\lEmbed{M_i}\in\fT$,
  we have proved that $\taylor M\inter \cone a$ is infinite.\qed
\end{proof}

Notice that the structure $\fS_{\text{sgl}}$ of singletons used in~\cite{Ehrhard10lics} does not enjoy the hypothesis of Theorem~\ref{theorem:finitary:implies:SN} ($\linearStructure\not\subset\fS_{\text{sgl}}$). In fact:
\begin{remark}\label{rk:counterexample_singletons}
We prove that $\taylor{\mathbf{\Omega_3}}\in\dual{\cones{\fS_{\text{sgl}}}}$, although $\mathbf{\Omega_3}$ is not normalizing. 
Recall from Example~\ref{ex:taylor_expansion}, that the support of the Taylor expansion of $\mathbf{\Omega_3}$ is made of terms of the form $\rappl{\delta_{n_0,m_0}}{\mset{\delta_{n_1,m_1},\dots,\delta_{n_k,m_k}}}$ (for $k,n_i,m_i\in\N$). Write 
$
	\Delta_h=\set{
		\rappl{\delta_{-,-}}{\mset{\dots\delta_{-,-}\dots}\cdots\mset{\dots\delta_{-,-}\dots}}\text{ with $h$ bags}
	}
$: in particular $\taylor{\mathbf{\Omega_3}}=\Delta_1$. One can easily check that if $s\in \Delta_h$ and $s\reducesTo s'$, then $s'\in\Delta_{h'}$ with $h\le h'$. A careful inspection of such reductions shows that they are moreover reversible: for all $s'\in\Delta_{h'}$ and all $h\le h'$ there is exactly one $s\in\Delta_h$ such that $s\reducesTo s'$. It follows that $\Delta_1\cap \uparrow s$ is either empty or a singleton. Therefore  $\taylor{\mathbf{\Omega_3}}\in\dual{\cones{\fS_{\text{sgl}}}}$.
\end{remark}

%% file: F-conclusion.tex
%

We achieved all implications of Figure~\ref{fig:main_results}, but the rightmost one, concerning the finiteness of the coefficients in the normal form of the Taylor expansion of a strongly normalizing $\lambda$-term (recall Equation~\eqref{eq:power_series_taylor} in the Introduction). 
%

Thanks to the definition of cones (Definition~\ref{def:cones}) we immediately have the following lemma, which is the last step to Corollary~\ref{cor:NF}.
\begin{lemma}
	\label{lemma:NF}
	Let $\fT$ be a structure. If $\taylor M\in\finitarySets[\fT]$, then $\forall t\in\bigcup\fT$, $\NF(\taylor M)_t$ is finite.
%
%
\end{lemma}


Applying Corollary~\ref{cor:realizability} and Lemma~\ref{lemma:NF} to a structure like $\boundedStructure$ or $\fS_{\text{sgl}}$, we get:
\begin{corollary}\label{cor:NF}
 Given a non-deterministic $\lambda$-term $M$, if $M$ is strongly normalizable, then $\NF(\taylor M)_t$ is finite for all $t\in\resourceTerms$.
\end{corollary}